%% file: nubots-agitation.tex
\documentclass[11pt]{article} %

\usepackage[]{geometry}  %
\usepackage{amsmath}
\usepackage{amssymb}
\usepackage{amsfonts}
\usepackage{amsthm}
\usepackage{graphicx}
\usepackage{color}
\usepackage{microtype}

\newcommand{\R}{\mathbb{R}}
\newcommand{\Z}{\mathbb{Z}}
\newcommand{\range}[1]{\{1,\ldots,#1\}}
\newcommand{\Exp}{\mathrm{E}}

\usepackage{thmtools, thm-restate}
\declaretheorem{theorem}

\newtheorem{lemma}[theorem]{Lemma}
\newtheorem{definition}[theorem]{Definition}
\usepackage[colorlinks=false]{hyperref}

\begin{document}

\title{Fast algorithmic self-assembly of simple shapes\\ using random agitation}

\newlength{\authorgap}
\setlength{\authorgap}{0.07\textwidth}
\author{%
 Ho-Lin Chen%
   \thanks{National Taiwan University.
   \protect\url{holinchen@ntu.edu.tw}.
   Supported by NSC grant 101-2221-E-002-122-MY3.}%
   \hspace{\authorgap}
   David Doty%
    \thanks{California Institute of Technology.
    \protect\url{ddoty@caltech.edu}.
    Supported by National Science Foundation grants 0832824 \& 1317694 (The Molecular Programming Project), CCF-1219274,  CCF-1162589.}
    \hspace{\authorgap}
Dhiraj Holden%
   \thanks{California Institute of Technology.
   \protect\url{dholden@caltech.edu}.
   Supported by NSF grant CCF-1219274.}
    \\
Chris Thachuk%
    \thanks{California Institute of Technology.
    \protect\url{thachuk@caltech.edu}.
     Supported by NSF grant CCF/HCC-1213127 \& a Banting Fellowship.
    }
    \hspace{\authorgap}
Damien Woods%
    \thanks{California Institute of Technology.
    \protect\url{woods@caltech.edu}.
    Supported by National Science Foundation grants 0832824 \& 1317694 (The Molecular Programming Project), CCF-1219274,  CCF-1162589.}
    \hspace{\authorgap}
   Chun-Tao Yang%
    \thanks{National Taiwan University.
    \protect\url{havachoice@gmail.com}.
    Supported by NSC grant 101-2221-E-002-122-MY3.}
    }
\date{}

\maketitle

\begin{abstract}
We study the power of uncontrolled random molecular movement in the nubot model of self-assembly.  The nubot model is an asynchronous nondeterministic cellular automaton augmented with rigid-body movement rules (push/pull, deterministically and programmatically applied to specific monomers) and random agitations (nondeterministically applied to every monomer and direction with equal probability all of the time).  Previous work on the nubot model showed how to build simple shapes such as lines and squares quickly---in expected time that is merely logarithmic of their size.  These results crucially make use of the programmable rigid-body movement rule: the ability for a single monomer to control the movement of a large objects quickly, and only at a time and place of the programmers' choosing.  However, in engineered molecular systems, molecular motion is largely uncontrolled and fundamentally random. This raises the question of whether similar results can be achieved in a more restrictive, and perhaps easier to justify, model where uncontrolled random movements, or agitations, are happening throughout the self-assembly process and are the only form of rigid-body movement.  We show that this is indeed the case: we give a polylogarithmic expected time construction for squares using agitation, and a sublinear expected time construction to build a line.  Such results are impossible in an agitation-free (and movement-free) setting and thus show the benefits of exploiting uncontrolled random movement.
\end{abstract}

\input{sections/intro}

\input{sections/model}

\input{sections/slowline_subsect}

\input{sections/sync_subsect}

\input{sections/squares}

\input{sections/line}

\section*{Acknowledgments}
A special thanks to Erik Winfree for many insightful and helpful discussions on the model and constructions. We also thank Robert Schweller, Matthew Cook and Andrew Winslow for discussions on the model and problems studied in this paper.

\bibliographystyle{abbrv} %
\bibliography{nubots}

 \newpage
 \appendix
 \section*{Appendix}

\input{sections/appendix-movement}

\end{document}

%% file: sections/intro.tex
\section{Introduction}
Every molecular structure that has been self-assembled in nature or in the lab was assembled in conditions (above absolute zero) where molecules are vibrating relative to each other, randomly bumping into each other via Brownian motion, and often experiencing rapid uncontrolled fluid flows. It makes sense then to study a model of self-assembly that includes, and indeed allows us to exploit and program, such phenomena. It is a primary goal of this paper to show the power of self-assembly under such conditions.

In the theory of molecular-scale self-assembly, millions of simple interacting components are designed to autonomously stick together to build complicated shapes and patterns. Many models of self-assembly are cellular automata-like crystal growth models, such as the abstract tile assembly model~\cite{Winf98thesis}. Indeed this and other such models have given rise to a rich theory of self-assembly~\cite{doty2012,patitz2012introduction,WoodsIU}. In biological systems we frequently see much more sophisticated growth processes, where self-assembly is combined with active molecular motors that have the ability to push and pull large structures around. For example, during the gastrulation phase of the embryonic development of the model organism \emph{Drosophila melanogaster} (a fly) large-scale (100s of micrometers) rearrangements of the embryo are effected by thousands of (nanoscale) molecular motors working together to rapidly push and pull the embryo into a final desired shape~\cite{dawes2005folded,Wieschaus08pulsed}.
We wish to understand, at a high level of abstraction, the ultimate computational capabilities and  limitations of such molecular scale rearrangement and growth. 

The nubot model of self-assembly, put forward in~\cite{nubots}, is  an asynchronous nondeterministic cellular automaton augmented with non-local  rigid-body movement. Unit-sized monomers are placed on a 2D hexagonal grid. Monomers can undergo  state changes, appear, and disappear, using local cellular-automata style rules. However, there is also a non-local aspect to the model, a kind of rigid body movement that comes in two forms: movement rules and random agitations.  
 A {\em movement rule} $r$, consisting of a pair of monomer states $A,B$ and two unit vectors, is a programatic way to specific unit-distance translation of a set of monomers in one step.   If $A$ and $B$ are in a prescribed orientation, one is nondeterministically chosen to move unit distance in a prescribed direction.  The rule $r$ is applied in a rigid-body fashion: if $A$ is to move right, it pushes anything immediately to its right and pulls any monomers that are bound to its left (roughly speaking) which in turn push and pull other monomers, all in one step. The rule may not be applicable if it is blocked (i.e.\ if movement of $A$ would force  $B$ to also move), which is analogous to the fact that an arm can not push its own shoulder.  The other form of movement in the model is called {\em agitation}: at every point in time, every  monomer on the grid may move unit distance in any of the six directions, at unit rate for each (monomer, direction) pair. An agitating monomer will push or pull any monomers that it is adjacent to, in a way that preserves rigid-body structure, all in one step. Unlike movement, agitations are never blocked.  Rules are applied asynchronously and in parallel in the model. Taking its time model from stochastic chemical kinetics, a nubot system evolves as a continuous time Markov process. 

In summary, there are two kinds of one-step parallel movement in the model:  (a) a {\em movement rule} is applied only to a pair of monomers with the prescribed states and orientation, and then causes the movement of one of these monomers along with  other pushed/pulled monomers, whereas (b) {\em agitations} are always applicable at every time instant, in every direction and to every monomer throughout the grid and an agitating monomer may push/pull other monomers.

In previous work, the movement rule was exploited to show that nubots are very efficient in terms of their computational ability to quickly build complicated shapes and patterns. Agitation was treated as something to be robust against (i.e.\ the constructions in~\cite{nubots,nubotsDNA19} work both with and without agitation), which seems like a natural requirement when building structures in a molecular-scale environment. However, it was left open as to whether the kind of results achieved with movement could be achieved without movement, but by exploiting agitation~\cite{nubotsDNA19}. 
In other words, it was left open as to whether augmenting a cellular automaton with an uncontrolled form of random rigid-body movement would facilitate functionality  that is impossible without it.  Here we show this is the case.

Agitation, and the movement rule, are defined in such a way that larger objects move faster, and this is justified by imagining that we are self-assembling rigid-body objects in a nanoscale environment where there is not only diffusion and Brownian motion but also convection, turbulent flow, cytoplasmic streaming and other uncontrolled inputs of energy interacting with each monomer in all directions. It remains as an interesting open research direction to look at the nubot model but with a slower rate model for agitation and movement, specifically where we hold on to the notion of rigid body movement and/or agitation but where bigger things move slower, as seen in Brownian motion for example.  Independent of the choice of rate model,  one of our main motivations here is to understand what can be done with {\em asynchronous, distributed and parallel} self-assembly with {\em rigid body motion}: the fact that our systems work in a parallel fashion is actually more important to us than the fact they are fast. It is precisely this engineering of distributed asynchronous molecular systems that interests us.  

The nubot model is related to, but distinct from, a number of other self-assembly and robotics models as described in~\cite{nubots}.  Besides the fact that biological systems make extensive use of molecular-scale movements and rearrangements, in recent years we have seen the design and fabrication of a number of molecular-scale DNA motors~\cite{turberfield07bat} and active self-assembly systems which also serve to motivate our work, details of which can be found in previous papers on nubots~\cite{nubots,nubotsDNA19}.

\subsection{Results and future work} 
 Let the {\em agitation nubot} model denote the nubot model without the movement rule and with agitation (see Section~\ref{sec:defs} for formal definitions). 
The first of our two main results shows that agitation can be exploited to build a large object exponentially quickly:
\begin{restatable}{theorem}{thmsquare}
\label{thm:square}%
There is a set of nubot rules $\mathcal{N}_{\mathrm{square}}$, such that for all $n  \in\mathbb{N}$, starting from a line of $\lfloor \log_2 n \rfloor +1$ monomers, each in state~$0$ or~$1$, $\mathcal{N}_{\mathrm{square}}$ in the agitation nubot model assembles an $n\times n$ square  in $O(\log^2 n)$ expected time, $n \times n$ space and $O(1)$ monomer states.\end{restatable}

The proof is in Section~\ref{sec:square}.
Our second main result shows that we can achieve sublinear  expected time growth of a length $n$ line in only $O(n)$ space:
\begin{restatable}{theorem}{thmline}
\label{thm:1D_line}%
There is a set of nubot rules $\mathcal{N}_{\mathrm{line}}$, such that for any $\epsilon > 0$, for sufficiently large $n \in\mathbb{N}$, starting from a line of $\lfloor \log_2 n \rfloor +1$ monomers, each in state~$0$ or~$1$, $\mathcal{N}_{\mathrm{line}}$ in the agitation nubot model assembles an $n\times 1$ line in $O(n^{1/3} \log n )$ expected time,  $n\times 5$ space and~$O(1)$ monomer states. 
\end{restatable}

The proof is in Section~\ref{sec:line}. 
Lines and squares are examples of fundamental components for the self-assembly of arbitrary computable shapes and patterns in the nubot model~\cite{nubots,nubotsDNA19,dabbyChenSODA2012} and other self-assembly models~\cite{doty2012,patitz2012introduction}. 

Our work here suggests that random agitations applied in an uncontrolled fashion throughout the grid are a powerful resource. However, are random agitations as powerful as the programable and more deterministic movement rule used in previous work on the nubot model~\cite{nubots,nubotsDNA19}?  In other words can agitation {\em simulate} movement?  More formally, is it the case that for each nubot program~$\mathcal{N}$, there is an agitation nubot program $\mathcal{A_N}$, that acts just like $\mathcal{N}$ but with some $m \times m$ scale-up in space, and a $k$ factor slowdown in time, where $m$ and $k$ are (constants)  independent of $\mathcal{N}$ and its input? This question is inspired by the use of simulations in tile assembly as a method to classify and separate the power of self-assembly systems, for more details see~\cite{IUSA,WoodsIU}. It would also be interesting to know whether the  full nubot model, and indeed the agitation nubot model, are intrinsically universal~\cite{IUSA,WoodsIU}. That is, is there a single set of nubot rules that simulate any nubot system? Is there a single set of agitation nubot rules that simulate any agitation nubot system? Here the scale factor $m$ would be a function of the number of monomer states of the simulated system $\mathcal{N}$. 
As noted in the introduction, it remains as an interesting open research direction to look at the nubot model but with a slower rate model for agitation and movement, as seen in Brownian motion, for example.

%% file: sections/model.tex
\newcommand{\config}{C}
\newcommand{\ra}{\rightarrow}
\newcommand{\la}{\leftarrow}
\newcommand{\La}{\Leftarrow}

\section{The nubot model}\label{sec:defs}
In this section we formally define the nubot model. Figure~\ref{fig:model} gives an overview of the model and rules, and Figure~\ref{fig:agitation_example} gives examples of agitation.  Figure~\ref{fig:slow_line} shows a simple example construction using only local rules. 

\begin{figure}[t]
  \centering
    \includegraphics[width=\textwidth]{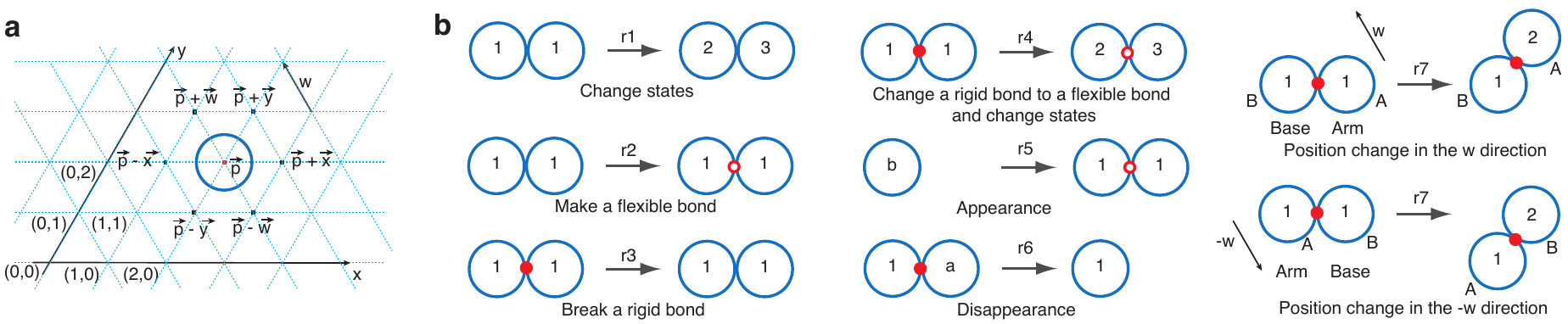}
  \caption{Overview of nubot model. (a) A nubot configuration showing a single nubot monomer on the triangular grid. (b) Examples of nubot monomer rules. Rules r1-r6 are local cellular automaton-like rules, whereas r7 effects a non-local movement. A flexible bond is depicted as an empty red circle and a rigid bond is depicted as a solid red disk. Rules and bonds are described more formally in Section~\ref{sec:defs}. Figure~\ref{fig:agitation_example} describes agitation.}
    \label{fig:model}
\end{figure}

The model uses a two-dimensional triangular grid with a coordinate system using axes $x$ and~$y$ as shown in Figure~\ref{fig:model}(a). 
A third axis, $w$, is defined as running through the origin and through $\overrightarrow{w} = -\overrightarrow{x} + \overrightarrow{y} = (-1,1)$, but we use only the~$x$ and~$y$ coordinates  to define position.
The \emph{axial directions} $\mathcal{D} =  \{ \pm\overrightarrow{x}, \pm\overrightarrow{y},  \pm\overrightarrow{w} \}$ are the unit vectors along  axes $x,y,w$. A grid point $\overrightarrow{p} \in \mathbb{Z}^2$ has the set of six \emph{neighbors}  $\{ \overrightarrow{p} + \overrightarrow{u} \mid \overrightarrow{u} \in \mathcal{D} \}$.  Let $S$ be a finite set of monomer states.
A nubot {\em  monomer} is a pair $X = (s_i, p(X)$) where $s_i \in S$ is a state and $p(X) \in \mathbb{Z}^2 $ is a grid point.
Two monomers on neighboring grid points are either connected by a \emph{flexible}  or  \emph{rigid} bond, or else have no bond (called a \emph{null} bond).
Bonds are described in more detail below.
A {\em configuration}~$C$ is a finite set of monomers along with all of the bonds between them (unless otherwise stated a configuration consists of {\em all} of the monomers on the grid and their bonds).

One configuration {\em transitions} to another either via the application of a \emph{rule} that acts on one or two monomers, or by an {\em agitation}. 
For a rule $r = (s1, s2, b, \overrightarrow{u}) \rightarrow (s1', s2', b', \overrightarrow{u}')$,  the left and right sides of the arrow respectively represent the contents of the two monomer positions before and after the application of~$r$.  Specifically, $s1, s2, s1', s2'  \in S \cup \{ \mathsf{empty} \}$ are monomer states where $\mathsf{empty}$ denotes lack of a monomer,
$b, b' \in \{\mathsf{flexible}, \mathsf{rigid}, \mathsf{null} \}$ are bond types, and $\overrightarrow{u}, \overrightarrow{u}' \in \mathcal{D}$ are unit vectors. 
$b$ is a bond type between monomers with state $s1$ and $s2$, and $\overrightarrow{u} \in \mathcal{D}$ is the relative position of a   monomer with state $s2$ to a  monomer with state  $s1$ (likewise for $b',s1',s2', \overrightarrow{u}'$).
At most one of $s1, s2$ is $ \mathsf{empty}$ (we disallow spontaneous generation of monomers from empty space).
If  $\mathsf{empty} \in \{ s1 , s2 \}$    then $b = \mathsf{null}$, likewise if $\mathsf{empty} \in \{ s1', s2' \}$    then $b' = \mathsf{null}$  (monomers can not be bonded to empty space).

A rule either does not or does involve movement (translation).  First, in the case of no movement we have $\overrightarrow{u} = \overrightarrow{u}'$. Thus we have a rule of the form $r = (s1, s2, b,\overrightarrow{u}) \rightarrow (s1', s2', b', \overrightarrow{u})$, where the monomer pair may  {\em change state} ($s1 \neq s1'$ and/or $s2 \neq s2'$ )  and/or {\em change bond} ($b \neq b'$),  examples are shown in Figure~\ref{fig:model}(b). If $s_i \in \{s1,s2\}$ is $\mathsf{empty}$ and $s_i'$ is not, then the rule is said to induce the \emph{appearance} of a new monomer  at the  empty location. If one or both monomer states go from non-empty to $\mathsf{empty}$, the rule induces the \emph{disappearance} of one or both  monomers.  Second, in the case of a movement rule, the rule has a specific form as defined in Appendix~\ref{sec:movement_definition}. Movement rules are not used in the {\em agitation nubot} model studied in this paper, and so their definition may be skipped by the reader. A rule is only applicable in the orientation specified by~$\overrightarrow{u}$.

\begin{figure}[t]
  \centering
    \includegraphics[width=\linewidth]{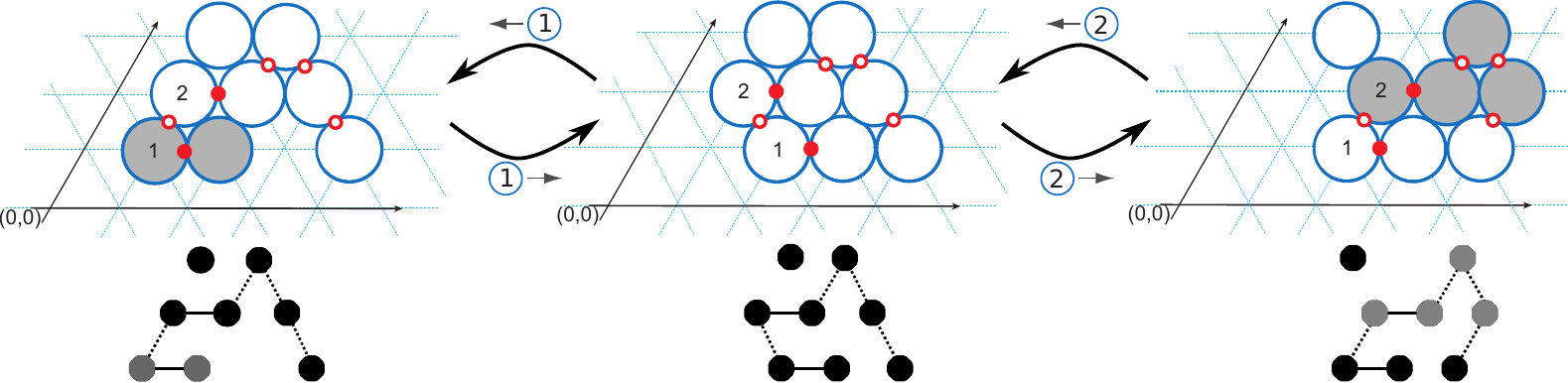}
  \caption{Top:  Example agitations. Starting from the centre
    configuration, there are~$48$ possible agitations (8 monomers, 6
    directions each), any one of which is chosen with equal
    probability $1/48$. The right configuration results from the
    agitation of the monomer at position $(1,2)$ in the direction
    $\rightarrow$, starting from the centre configuration.  The left
    configuration results from the agitation of the monomer at
    position $(2,1)$ in the direction $\leftarrow$, starting from the
    centre configuration.  The shaded monomers are the {\em agitation
      set}---the set of monomers that are moved by the agitation---when beginning from the centre configuration. Bottom:  simplified ball-and-stick representation of the monomers and their bonds, which is used in a number of other figures.}
  \label{fig:agitation_example}
\end{figure}

To define agitation we introduce some notions. 
Let $\overrightarrow{v} \in \mathcal{D}$ be a unit vector.
The $\overrightarrow{v}$-boundary of a set of monomers $S$ is defined to be the set of grid points outside of $S$  that are unit distance in the $\overrightarrow{v}$ direction from monomers in~$S$.

\begin{definition}[Agitation set]\label{def:agitationset}
Let  $\config$ be a configuration containing  monomer~$A$, and let $\overrightarrow{v} \in \mathcal{D}$ be a unit vector. The \emph{agitation set} $\mathcal{A}(\config, A, \overrightarrow{v})$ is defined to be the smallest {monomer} set in $\config$ containing $A$ that can be translated by $\overrightarrow{v}$ such that: (a)  monomer pairs in $\config$ that are joined by rigid bonds do not change their relative position to each other, (b) monomer pairs in $\config$ that are joined by flexible bonds stay within each other's neighborhood, and (c)  the $\overrightarrow{v}$-boundary of $\mathcal{A}(\config, A, \overrightarrow{v})$ contains no monomers.
\end{definition}

We now define agitation. An \emph{agitation} step acts on an entire configuration $\config$ as follows. A monomer~$A$ and unit vector $\overrightarrow{v} $ are selected uniformly at random from the configuration of monomers $\config$ and the set of six unit vectors $\mathcal{D}$ respectively. Then, the agitation set  $\mathcal{A}(\config, A, \overrightarrow{v})$ of monomers  (Definition~\ref{def:agitationset})  moves  by vector~$\overrightarrow{v}$.

Figure~\ref{fig:agitation_example} gives two examples of agitation. Some remarks on agitation: It can be seen that for any non-empty configuration the agitation set is always non-empty.   During agitation, the only change in the system configuration is in the positions of the constituent monomers in the agitation set, and all of the monomers' states and bond types remain unchanged.  We let the {\em agitation nubot} model be the nubot model without the movement rule. Agitation is intended to model movement that is not a direct consequence of a rule application, but rather results from diffusion, Brownian motion,  turbulent flow or other uncontrolled inputs of energy. 

A \emph{nubot system} $\mathcal{N} = (C_0, \mathcal{R})$ is a pair  where $C_0$ is the initial configuration, and $\mathcal{R}$ is the set of  rules. If configuration $C_i$  transitions to $C_j$ by some  rule $r \in \mathcal{R}$, or by an agitation step, we write $C_i \vdash  C_j$. A {\em trajectory} is a finite sequence of configurations $C_1, C_2, \ldots , C_\ell$ where  $C_i \vdash  C_{i+1}$ and $1 \leq i \leq \ell-1$.  A nubot system is said to {\em assemble} a target configuration $C_t$ if, starting from the initial configuration~$C_0$, every trajectory evolves to a  translation of $C_t$. 

A nubot system evolves as a continuous time Markov process. The rate for each rule application, and for each  agitation step, is 1. If there are $k$ applicable transitions for a configuration~$C_i$ (i.e.\ $k$ is the sum of the number of rule and agitation  steps that can be applied to all monomers), then the probability of any given transition being applied is $1/k$, and the time until the next transition is applied is  an exponential random variable with rate $k$ (i.e.\ the expected time is $1/k$).  The probability of a trajectory is then the product of the probabilities of each of the transitions along the trajectory, and the expected time of a trajectory is the sum of the expected times of each transition in the trajectory. Thus, $\sum_{t \in \mathcal{T}} \mathrm{Pr}[t] \cdot  \mathrm{time}(t)$ is the expected time for the system to evolve from configuration~$C_i$ to configuration~$C_j$, where~$\mathcal{T}$ is the set of all trajectories from~$C_i$ to  any configuration isomorphic (up to translation and agitation) to $C_j$, that do not pass through any other configuration isomorphic to~$C_j$,  and $ \mathrm{time}(t)$ is the expected time for trajectory $t$.

The complexity measure {\em number of monomers} is the maximum number of monomers that appears in any configuration. The {\em number of states} is the total number of distinct monomer states that appear in the rule set. {\em Space} is the maximum area, over the set of all reachable configurations, of the minimum area $l \times w$ rectangle (on the triangular grid) that, up to translation, contains all monomers in the configuration.

%% file: sections/slowline_subsect.tex
\subsection{Example: A simple, but  slow, method to build a line}
Figure~\ref{fig:slow_line}, taken from~\cite{nubots}, shows a simple method to build a length~$n$ line in expected time~$n$, using~$O(n)$ monomer states. Here, the program is acting as an asynchronous cellular automata and is {\em not} exploiting the ability of a large set of monomers to quickly move via agitation. Our results show that by using agitation we can do much better than this very slow and expensive (many states) method to grow a line.

  \begin{figure}[t] %
    \begin{center}
      \includegraphics[width = \textwidth]{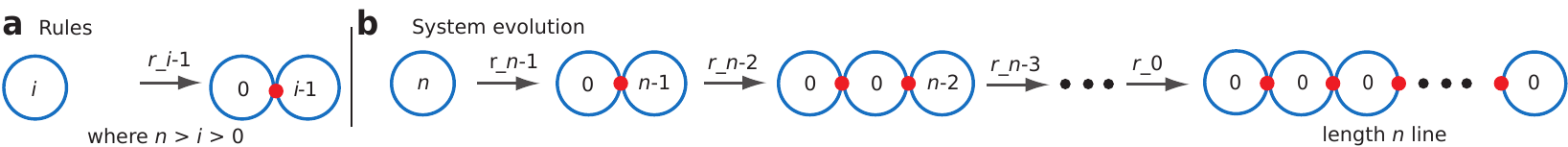}
      \caption{A nubot system that slowly grows a length~$n$ line in $O(n)$ time,~$n$ monomer states, and using space $n \times 1$.
(a)~Rule set: $\mathcal{R}_n^{\textrm{slow line}} = \{ r_i \mid r_i = (i, \mathsf{empty}, \mathsf{null}, \vec{x}) \ra (0, i-1, \mathsf{rigid}, \vec{x}), \mathrm{~where~}  n >  i > 0\}$. (b)~Starting from an initial configuration with a single monomer in state $n$, the system generates a length $n$ line. Taken from~\cite{nubots}.}
      \label{fig:slow_line}
    \end{center}
  \end{figure}

%% file: sections/sync_subsect.tex
\section{Synchronization via agitation}\label{sec:sync}
In this section we describe a fast method that uses agitation to synchronize the states of a line of monomers, or in other words, to reach consensus. Specifically, the {\em synchronization} problem is: given a length-$m$ line of monomers that are in a variety of states but that all eventually reach some target state $s$, then after all $m$ monomers have reached state $s$, communicate this fact to all $m$ monomers in $O(\log m)$ expected time. 

\begin{figure}%
  \centering
    \includegraphics[width=\textwidth]{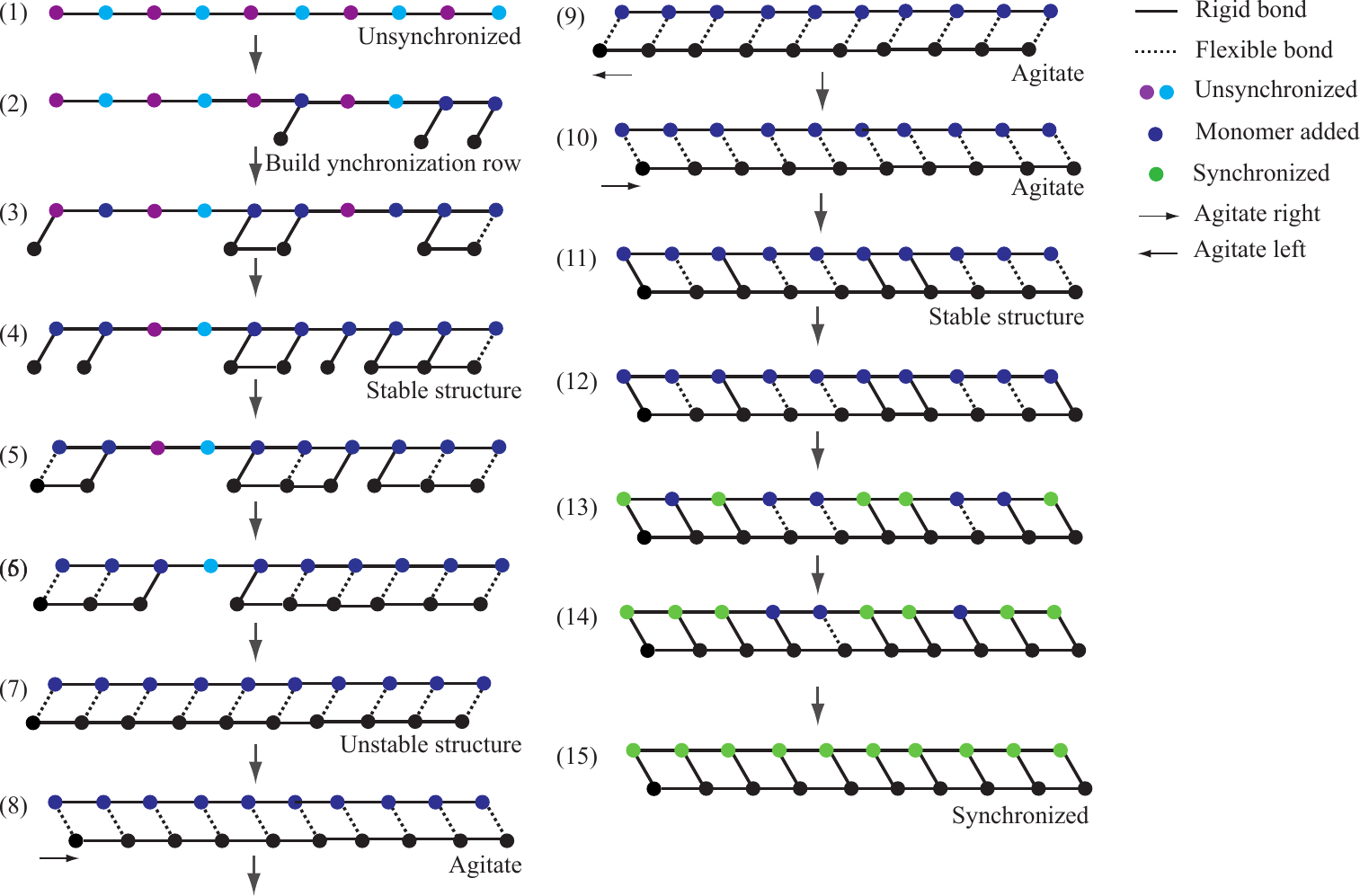}
  \caption{Synchronization via agitation: a nubot construction to synchronize (or send a signal, or reach consensus) between $n$ monomers in $O(\log n)$ expected time.  Steps (1)--(6): build a row of monomers called the synchronization row. Rigid bonds are converted to flexible bonds in such a way that agitations do not change the relative position of monomers. A structure with this property is said to be  {\em stable}. Specifically, monomers are added using rigid vertical  bonds; new monomers join to left-right neighbours using rigid horizontal bonds; when a monomer is bound horizontally to both neighbours it makes its {\em vertical} bond flexible; monomers on the extreme left and right of the synchronization row are treated differently---their vertical bonds become flexible after joining any horizontal neighbour. This enforces that the entire structure is stable up until the final horizontal bond is added, and then the structure becomes unstable in such a way that the synchronization row can agitate left-right relative to the backbone row. Steps (7)--(10), the structure is not stable, and the synchronization row is free to agitate left and right relative to the backbone row.  While agitating, the synchronization row spends half the time to the left, and half to the right, of the backbone row. However, whenever the synchronization row is to the right a rigid bond may form between any synchronization row monomer and the backbone monomer directly above, hence the first such bond forms in expected time~$1/m$, where~$m$ is the length of the backbone. Then all bonds become rigid in~$O(\log m)$ expected time, during which time (12)--(15) the backbone monomers change their state to the final {\em synchronized} state.}
    \label{fig:sync}
\end{figure}

\begin{lemma}[Synchronization]\label{lem:sync_agitate}
A line of monomers of length $m \in \mathbb{N}$ can be synchronized (all monomers put into the same state) in $O(\log m)$ expected time,  $m \times 2$ space and~$O(1)$~states. 
\end{lemma}
The proof is described in Figure~\ref{fig:sync} and its caption.  The figure gives a synchronization routine that is used throughout our constructions. This is a modification of the synchronization routine in~\cite{nubots}, made to work with agitation instead of the movement rule.

%% file: sections/squares.tex
\begin{figure}[t]
  \centering
  \includegraphics[width=\textwidth]{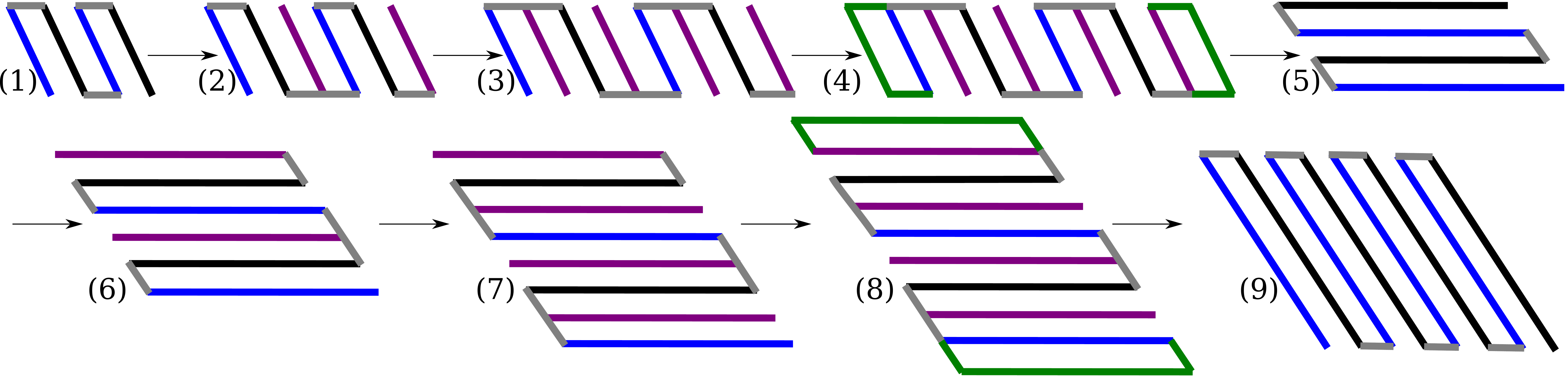}
\caption{An overview of the square doubling algorithm that grows an $m \times m$ zig-zag ``comb'' to a $2m \times 2m$ comb. (1) An initial $m \times m$ comb with vertical teeth, is (2)  ``half-doubled'' to give a $\lfloor1.5m\rfloor \times m$ comb, which is (3) again half-doubled to give a $2m \times m$ comb.   (4)--(5) The internal bond structure is reconfigured to give a comb with horizontal teeth. (6)--(7) this comb is vertically doubled in size and then (8)--(9) reconfigured to give a $2m \times 2m$ comb with vertical teeth.  The green
  lines indicate temporary synchronization rows that are used when
  reorientating the teeth of the comb.}\label{fig:squareoverview}
\end{figure}

\begin{figure}[t]
  \centering
    \includegraphics[width=\textwidth]{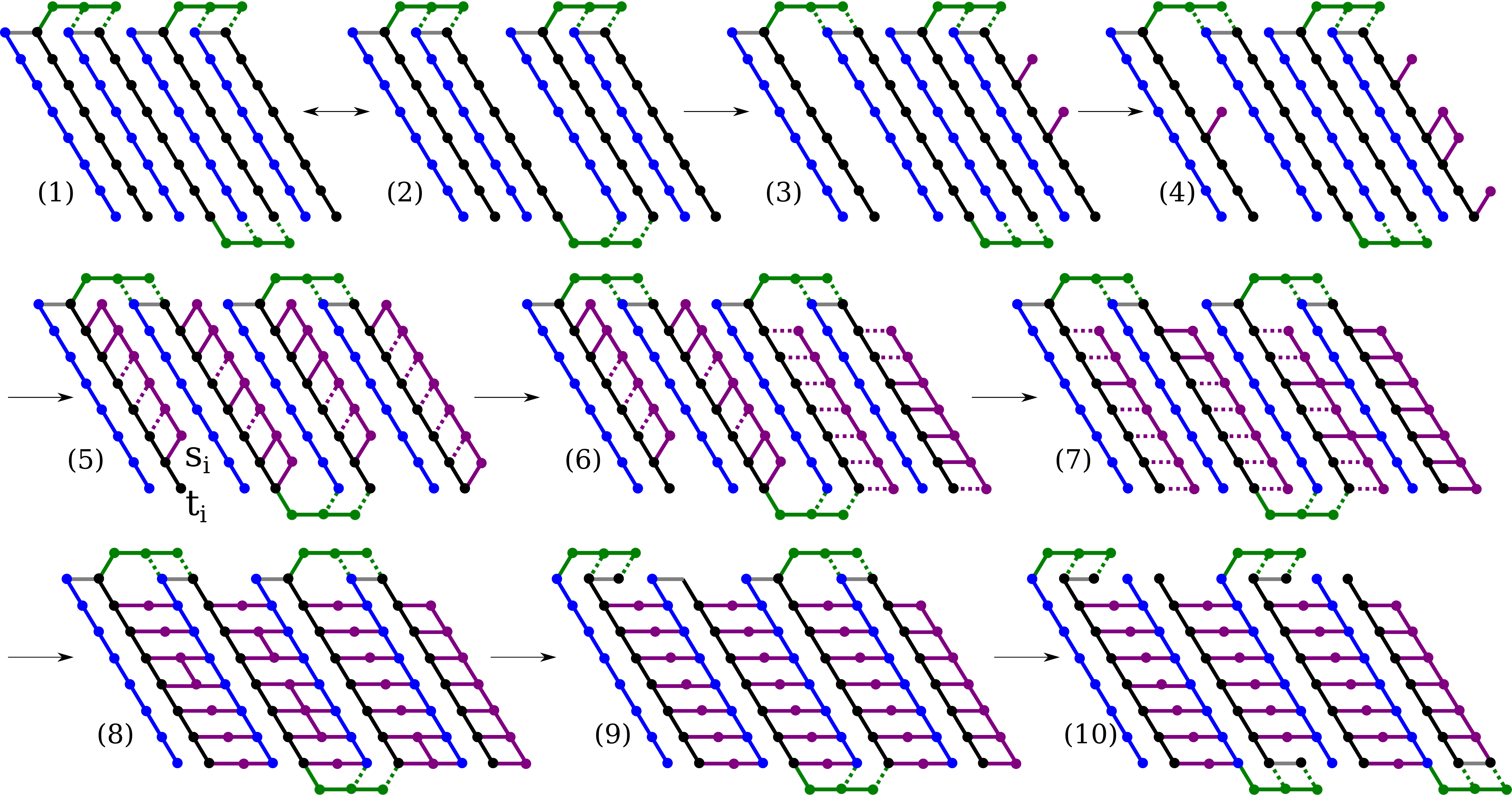}
\caption{The $m \times m$ to $\lfloor 1.5m \rfloor  \times m$ horizontal half-doubling algorithm, for $m=8$. This shows the details for step (1) to (2) of Figure~\ref{fig:squareoverview}. Monomer states are denoted using colours (bonds are also coloured for readability). Rigid bonds are solid, flexile bonds are dotted.  See main text for details.}\label{fig:squaredoubling}
\end{figure}

\section{Building squares via agitation}\label{sec:square}
This section contains the proof of our first main result,  Theorem~\ref{thm:square}, which we restate here:

\thmsquare*

\begin{proof}%
\noindent{\bf Overview of construction. \ } 
Figure~\ref{fig:squareoverview} gives an overview of our construction. A binary string that represents $n \in \mathbb{N}$ in the standard way is encoded as a string $x$, of length $\ell = \lfloor \log_2 n \rfloor +1$, of adjacent rigidly bound binary nubot monomers (each in state 0 or 1) placed somewhere on the hexagonal grid.

The leftmost of these monomers begins an iterated square-doubling process, that happens exactly $\ell$ times.  Each iteration of this square-doubling process: reads the current most significant bit $x_i$ of $x$, where $0\leq i \leq \ell $,  stores it in the state of a monomer in the top-left of the square and then deletes $x_i$. Then, if $x_i=0$ it takes an $m \times m$ comb
structure and doubles its size to give a $2m \times 2m$ comb structure, or if $x_i=1$ it gives a $(2m+1) \times (2m+1)$ structure. We will prove that each square-doubling step takes $O(\log m)$ time. There are $\ell$ rounds of square-doubling, i.e.\ the number of input monomers $\ell$ act as a counter to control the number of iterations, and since $m\leq n$ throughout, the process completes in the claimed expected time of $O(\log^2 n)$. The main part of the construction, detailed below, lies in the details of how each doubling step works and an expected time analysis, and constitutes the remainder of the proof.

\noindent{\bf Square-doubling. \ }
A single square-doubling
consists of four phases: two horizontal ``half-doublings'' and two
vertical half-doublings. Figure~\ref{fig:squareoverview} gives an overview. Figure~\ref{fig:squaredoubling} gives the
details of how we do the first of two horizontal half-doublings;  more
precisely, the figure shows how to go from an $m\times m$ structure to
a structure of size $\lfloor 1.5m\rfloor\times m$. Assume we are at a configuration
with $m$ vertical comb teeth (Figure~\ref{fig:squaredoubling}(1)) each
of height $m$ (plus some additional monomers). Teeth are numbered from
the left $t_1, t_2, \ldots , t_m$.  Each tooth monomer undergoes
agitation. It can be seen in Figure~\ref{fig:squaredoubling}(1)--(4), from the bond structure, that the only
agitations that change the relative position of monomers are left or
right agitations which move the green flexible bonds (depicted as
dashed lines)---all other agitations move the entire structure
without changing the relative positions of any monomers. Furthermore,
left-right monomer
agitations can create gaps between teeth~$t_i$ and~$t_{1+1}$ for even
$i$ only---for odd~$i$, teeth $t_i$ and $t_{1+1}$ are rigidly
bound.  An example of a gap opening between tooth $t_4$ and tooth $t_5$ is
shown in Figure~\ref{fig:squaredoubling}(2).
If a gap
appears between teeth $t_i$ and $t_{1+1}$ then each of the $m$ monomers in tooth~$t_i$ tries to attach a new purple monomer to its right (with a rigid bond, and each at rate 1), so attachment for any monomer to tooth~$i$ happens at rate~$m$. (Note that the gap is closing and opening at some rate also---details
in the time analysis.) After the first such purple monomer appears, the gap $g_i$, to the
right of tooth $t_i$, is said to be ``initially filled''.  For
example, in Figure~\ref{fig:squaredoubling}(4), gap $g_2$ is initially filled.

When gaps appear between teeth monomers, and then become initially filled,
additional monomers are attached, asynchronously and in
parallel. Monomers attaching to tooth $t_i$ initially attach by rigid
bonds as shown in
Figure~\ref{fig:squaredoubling}(4). As new monomers attach to
$t_i$, they then attempt to bind to each other vertically, and after
such a binding event they undergo a sequence of bond
changes---see
Figure~\ref{fig:squaredoubling}(4)-(9). Specifically, let $s_{i,j}$ be the $j^{\mathrm{th}}$ monomer
on the newly-forming ``synchronization row'' $s_i$ adjacent to $t_i$. When the neighbors
$s_{i,j-1}, s_{i,j+1}$ of monomer $s_{i,j}$ appear, then $s_{i,j}$
forms rigid bonds with them (at rate 1). After this, $s_{i,j}$ changes
its rigid bonds to $t_{i,j}$ to flexible. The top and bottom monomers
$s_{i,1}$, $s_{i,m}$ are special cases: their bonds to $t_{i,1}$,
$t_{i,m}$ become flexible after they have joined to their (single)
neighbors $s_{i,2}$, $s_{i,m-1}$.  Changing bonds in this order
guarantees that only {\em after all monomers} of $s_i$ have attached,
and not before, the synchronization row $s_i$ is free to agitate up
and down relative to the tooth~$t_i$ (this is the same technique for
building a synchronization row as described in
Section~\ref{sec:sync}). The new vertical synchronization row $s_i$
is then free to agitate up and down relative to its left-adjacent
tooth~$t_i$. When $s_{i,j}$ is ``down'' relative to~$t_{i,j}$ the
horizontal bonds between~$s_{i,j}$ and~$t_{i,j}$ become rigid, at rate
1 per bond (Figure~\ref{fig:squaredoubling}(6)--(7)).  %
When the vertical synchronization of $s_i$ is done, a message is sent from the top monomer $t_{i,m}$
of~$t_i$ (after its bond to $s_{i,m}$ becomes rigid) to the adjacent
monomer at the top of the comb. This results in the formation of a
horizontal synchronization row at the top of the structure. Using a
similar technique, a horizontal synchronization row grows at the
bottom of the structure.  After all~$2\lfloor 0.5 m\rfloor $ such messages have arrived,  and not before, the horizontal synchronization rows at the top and
bottom of the (now) $\lfloor 1.5m\rfloor \times m$ comb change the last of their
rigid (vertical) bonds to flexible and those synchronization rows are
free to agitate left/right and then lock into position, signaling to all monomers
along their backbone that the first of the four half-doublings of the
comb has finished.

The system prepares for the next horizontal half-doubling which will
grow the $\lfloor 1.5m\rfloor \times m$ comb to be an $2m \times m$
comb. The bonds at the top and bottom horizontal synchronization rows
reconfigure themselves (preserving connectivity of the overall
structure---see the description of reconfiguration below) in such a
way as to build the gadgets needed for the next half-doubling.  (Specifically, we want to now double teeth $t_i$ for odd $i \leq m$.) The construction proceeds similarly to the first half-doubling, except for the following change. After tooth synchronization row $s_1$ has synchronized, tooth $t_1$ grows a vertical synchronization row to its left, and after $s_{m}$ has synchronized, tooth $t_{m}$ grows a vertical synchronization row to its right (Figure~\ref{fig:squareoverview}(4)). These two synchronization rows are used to set-up the bond structure for the next stage of the construction (where we will reconfigure the entire comb so that the teeth are horizontal).

This covers the case of the input bit being~0. Otherwise, if the input bit
is~1, adding an extra tooth can be done using the single vertical
synchronization row on the right---it reconfigures
itself to have the bond structure of a tooth and then grows a new
vertical synchronization row.

\noindent{\bf Reconfiguration. \ } Next we describe how the
comb with vertical teeth is reconfigured to have horizontal teeth, as in
Figure~\ref{fig:squareoverview}(4)--(5). After synchronization row $s_i$ has
synchronized, each monomer $s_{i,j}$ in $s_i$ already has a rigid
horizontal bond to monomer $t_{i,j}$. After both $s_i$ and $s_{i+1}$
have synchronized, for all~$j$, monomers $s_{i,j}$ and $t_{i+1,j}$
bond using a horizontal rigid bond (at rate~1) for each pair $(s_{i,j} ,
t_{i+1,j}$). Monomers $t_i$ and $s_i$ then delete their vertical rigid
bonds in such a way that preserves the overall connectivity of the
structure. (For these bond reconfigurations we are simply using local---asynchronous cellular automaton style---rules that  preserves  connectivity. This trick has been used in previous nubot constructions in Section 6.5 of~\cite{nubots} and in~\cite{nubotsDNA19}.) This leads to a bond structure similar to that in Figure~\ref{fig:squaredoubling}(10) both with roughly twice the number of horizontal purple bonds: i.e.\ for each $j$,  $1\leq j\leq m$, there is now a horizontal straight line of purple bonds from the $j$th monomer on the leftmost vertical line to the $j$th monomer on the rightmost vertical line.  While this reconfiguration is taking place, the leftmost
and rightmost  vertical synchronization rows synchronize and delete
themselves, leaving appropriate gadgets to connect the horizontal
teeth: this signals the beginning of the next  two half-doubling
steps.

\noindent{\bf Expected time, space and states analysis.}
 Lemma~\ref{lem:doubling_analysis} states that the expected time to
perform a half-doubling is $O(\log m)$ for an $m \times m$ comb, and
since $n\leq m$, the slowest half-doubling takes expected time $O(\log
n)$. Each doubling involves 2 horizontal half-doubling phases, and 2
vertical half-doubling phases, and the 4 phases are separated by discrete
synchronization events. Reconfiguration involves $O(n^2)$ bond and
state change events, that take place independently and in parallel
($O(\log n)$ expected time) as well as a constant number of
synchronizations that each take $O(\log n)$ expected time. Hence for $4( \lfloor \log_2 n \rfloor +1)$ such half-doublings, plus $\lfloor \log_2 n \rfloor +1$ reconfigurations, we get an overall expected time of~$O(\log^2 n)$.

 We've sketched how to make an $n\times n$ structure in  $(n+2) \times (n+2)$ space. To make the construction work in $n\times n$ space, we first subtract~2 from the input, and build an $(n-2) \times (n-2)$ structure, and then at the final step have the leftmost and rightmost horizontal, and topmost and bottommost vertical, synchronization rows become rigid and be the border of the final $n\times n$ structure. A final monomer is added on the top left corner and we are done.  By stepping through the  construction it can be seen that $O(1)$ monomer states are sufficient.\end{proof}

Intuitively, the following lemma holds because the long (length $m$) teeth allow for rapid, $O(1)$ time per tooth, and parallel insertion of monomers to expand the width of the comb. This intuition is complicated by the fact that teeth agitating open and closed may temporarily block other teeth inserting a new monomer. However, after an insertion actually happens further growth occurs independently and in parallel, taking logarithmic expected time overall.

\newcommand{\lemDoublingAnalysis}{A comb with $m$ teeth where each tooth is of height $m$, can be horizontally half-doubled to length $\lfloor 1.5m\rfloor$ in expected time $O(\log m)$ in the agitation nubot model.}
\begin{lemma}\label{lem:doubling_analysis}
\lemDoublingAnalysis
\end{lemma}

\begin{proof}

\newcommand{\open}{{\tt open}}
\newcommand{\closed}{{\tt closed}}
\newcommand{\filled}{{\tt initially filled}}

\newcommand{\oldanalysis}{Consider tooth $i$, where $1\leq  i \leq m$ for $i$ even.
A tooth can be \open{}, \closed{} or \filled{} (one new monomer inserted).
Although the remaining structure can affect the transition probabilities relevant to tooth $i$, in any state, the rate at which the tooth transitions from \closed{} to \open{} is at least $m$, the rate that it transitions from \open{} to \closed{} is at most $m^2$, and the rate at which it transitions from \open{} to \filled{} is exactly $m$.
We define a new Markov process, with states \open{}, \closed{}, and \filled{} and the transition probabilities just described, which is easier to analyze.
Clearly, the random variable representing the time for this process to transition from \closed{} to \filled{} upper bounds the random variable representing the time for the real nubots process to do the same for a single tooth.
We now show that this new random variable has expected value $O(1)$.

Let $T_{\mathrm{CF}}$ be the random variable representing the time to go from \closed{} to \filled{}.
Let $T_{\mathrm{CO}}$ be the random variable representing the time to go from \closed{} to \open{}.
Let $T_{\mathrm{OC}}$ be the random variable representing the time to go from \open{} to \closed{}, conditioned on that transition happening, and define $T_{\mathrm{OF}}$ similarly for going from \open{} to \filled{}.
Note that $\Exp[T_{\mathrm{CO}}] \leq \frac{1}{m}$, $\Exp[T_{\mathrm{OC}}] \geq \frac{1}{m^2}$, and $\Exp[T_{\mathrm{OF}}] = \frac{1}{m}$. Let $E_i$ represent the event that the process revisits state \closed{} exactly $i$ times after being in state \open{} (and before reaching state \filled{}).
Let $T_i$ be the random variable representing the time to take exactly $i$ cycles between the states \open{} and \closed{}.
Let $C$ be the random variable representing the number of cycles taken between the states \open{} and \closed{} before transitioning to state \filled{}.
Each time the process is in state \open{}, it has probability \dwm{``at least''; --- here is where we use that the rate from open to closed  is $\leq m^2$  } $\frac{1}{m}$ to go to state \filled{}, so $C$ is a geometric random variable with $\Exp[C] = m$.\dwm{$\Exp[C] \leq m$} 
Then  \dwm{looks like $1/{m^2}$ should be changed to $1/m$ in the second line. i.e. we state that $\Exp[T_{\mathrm{OC}}] \geq \frac{1}{m^2}$, but we are claiming to put an upper bound on the expected time, so really here we want to use the fact that $\Exp[T_{\mathrm{OC}}] \leq \frac{1}{m}$ (here we use that the rate from open to closed is $\geq m$---this would beed to be added in text above). This results in the final expression being $\leq 3$, which implies changing more text in the equations, as well as before, and after.  }\dwm{note to self: in one place we use $\Exp[T_{\mathrm{OC}}] \geq \frac{1}{m^2}$ and in another we use $\Exp[T_{\mathrm{OC}}] \leq \frac{1}{m}$}
\begin{eqnarray*}
  \Exp[T_{\mathrm{CF}}]
  &=&
  \Exp[T_{\mathrm{CO}}] + \Exp[T_{\mathrm{OF}}] + \sum_{i=0}^\infty \Pr[E_i] \cdot \Exp[T_i]
  \\&\leq&
  \frac{2}{m} + \sum_{i=0}^\infty \Pr[E_i] \cdot i \cdot \left(\frac{1}{m^2} + \frac{1}{m}\right)
  \\&=&
  \frac{2}{m} + \left(\frac{1}{m^2} + \frac{1}{m}\right) \sum_{i=0}^\infty \Pr[E_i] \cdot i
  \\&=&
  \frac{2}{m} + \left(\frac{1}{m^2} + \frac{1}{m}\right) \Exp[C]
  \\&=&
  \frac{2}{m} + \left(\frac{1}{m^2} + \frac{1}{m}\right) m
  \leq
  2.
\end{eqnarray*}
By Markov's inequality, the probability is at most $\frac{1}{2}$ that it will take more than time 4 to reach from \closed{} to \filled{}.
Because of the memoryless property of the Markov process, conditioned on the fact that time $t$ has elapsed without reaching state \filled{}, the probability is at most $\frac{1}{2}$ that it will take more than $t+4$ time to reach state \filled{}.
Hence for any $t > 0$, the probability that it will take more than than $4t$ time to reach from state \closed{} to \filled{} is at most $2^{-t}$.

Since this tail probability decreases exponentially, it follows that for $m/2$ teeth, the expected time for all of them to reach state \filled{} is~$O(\log m)$.\end{proof}
}  %

Consider tooth $i$, where $1\leq  i \leq m$ for $i$ even.
A tooth can be \open{}, \closed{} or \filled{} (one new monomer inserted).
Although the remaining structure can affect the transition probabilities relevant to tooth $i$, in any state, the rate at which the tooth transitions from \closed{} to \open{} is at least $m$, the rate that it transitions from \open{} to \closed{} is at least $m$ and at most $m^2$, and the rate at which it transitions from \open{} to \filled{} is exactly $m$.
We define a new simpler Markov process, with states \open{}, \closed{}, and \filled{} and the transition probabilities just described, which is easier to analyze than the underlying full process.
Clearly, the random variable representing the time for the new simpler process to transition from \closed{} to \filled{} upper bounds the random variable representing the time for the underlying full nubot process to do the same for a single tooth.
We now show that this random variable has expected value~$O(1)$.

Let $T_{\mathrm{CF}}$ be the random variable representing the time to go from \closed{} to \filled{}.
Let $T_{\mathrm{CO}}$ be the random variable representing the time to go from \closed{} to \open{}.
Let $T_{\mathrm{OC}}$ be the random variable representing the time to go from \open{} to \closed{}, conditioned on that transition happening, and define $T_{\mathrm{OF}}$ similarly for going from \open{} to \filled{}.
Note that $\Exp[T_{\mathrm{CO}}] \leq \frac{1}{m}$, $ \frac{1}{m^2} \leq \Exp[T_{\mathrm{OC}}] \leq \frac{1}{m} $, and $\Exp[T_{\mathrm{OF}}] = \frac{1}{m}$. Let~$E_i$ represent the event that the process revisits state \closed{} exactly $i$ times after being in state \open{}, and immediately before reaching state \filled{}.
Let $T_i$ be the random variable representing the time to take exactly $i$ cycles between the states \open{} and \closed{}.
Let $C$ be the random variable representing the number of cycles taken between the states \open{} and \closed{} before transitioning to state \filled{}.
Each time the process is in state \open{}, independently of how many cycles have happened (memoryless), it has probability $\geq \frac{1}{m}$ to go to state \filled{}, so $C$ is upper-bounded by a geometric random variable with $\Exp[C] \leq m$.  
Then  
\begin{eqnarray*}
  \Exp[T_{\mathrm{CF}}]
  &=&
  \Exp[T_{\mathrm{CO}}] + \Exp[T_{\mathrm{OF}}] + \sum_{i=0}^\infty \Pr[E_i] \cdot \Exp[T_i]
  \\&\leq&
  \frac{2}{m} + \sum_{i=0}^\infty \Pr[E_i] \cdot \Exp[T_i]
\end{eqnarray*}
and since $\Exp[T_{\mathrm{CO}}] \leq \frac{1}{m}$ and $ \Exp[T_{\mathrm{OC}}] \leq \frac{1}{m} $   we can substitute for $\Exp[T_i]$
\begin{eqnarray*}
  \Exp[T_{\mathrm{CF}}]
  &\leq&
  \frac{2}{m} + \sum_{i=0}^\infty \Pr[E_i] \cdot i \cdot \left(\frac{1}{m} + \frac{1}{m}\right)
  \\&=&
  \frac{2}{m} + \left(\frac{1}{m} + \frac{1}{m}\right) \sum_{i=0}^\infty \Pr[E_i] \cdot i
  \\&=&
  \frac{2}{m} + \left(\frac{2}{m}\right) \Exp[C]
  \\&=&
  \frac{2}{m} + \left(\frac{2}{m}\right) m
  \leq
  3 = O(1).
\end{eqnarray*}
By Markov's inequality, the probability is at most $\frac{1}{2}$ that it will take more than time 6 to reach from \closed{} to \filled{}.
Because of the memoryless property of the Markov process, conditioned on the fact that time $t$ has elapsed without reaching state \filled{}, the probability is at most $\frac{1}{2}$ that it will take more than $t+6$ time to reach state \filled{}.
Hence for any $t > 0$, the probability that it will take more than $6t$ time to reach from state \closed{} to \filled{} is at most $2^{-t}$.

Since this tail probability decreases exponentially, it follows that for $m/2$ teeth, the expected time for all of them to reach state \filled{} is~$O(\log m)$.\end{proof}

%% file: sections/line.tex
\section{Building lines via agitation}\label{sec:line}
In this section we prove our second main theorem, Theorem~\ref{thm:1D_line}.
We prove this by giving a line construction that works in merely $n \times 5 = O(n)$ space while achieving sublinear expected time $O(n^{1/3} \log n)$, and $O(1)$ monomer states.

\thmline*

\begin{figure}
  \centering
    \includegraphics[width=\textwidth]{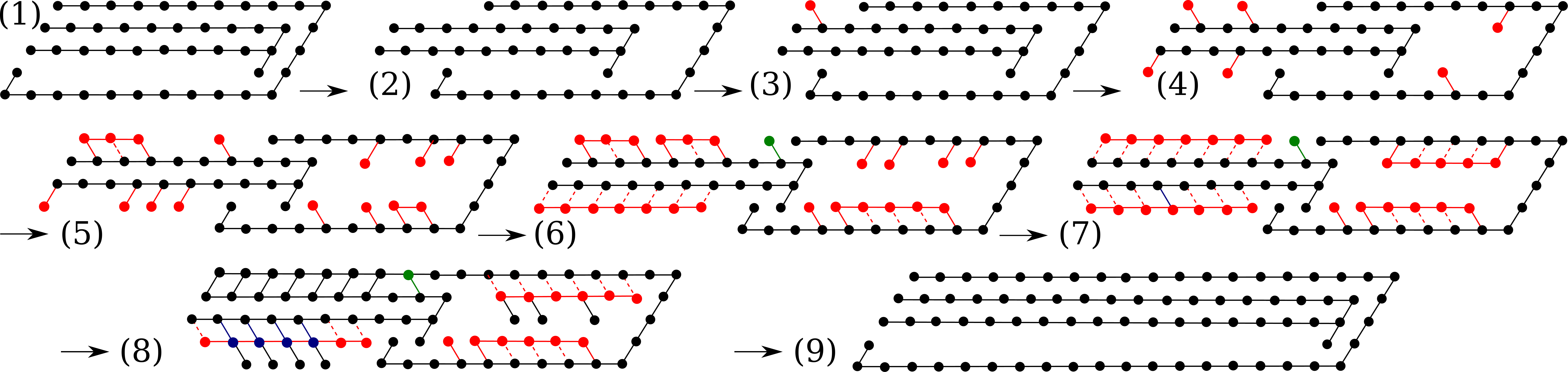}
\caption{Line doubling construction. The inner component is called the {\em sword}, which agitates left/right relative to the outer component called the scabbard (both are in black). The black sword-and-scabbard are doubled from length $m=8$ to length $2m=16$. Other monomers (red, green, blue) serve to both ratchet the movement, and to quickly in parallel build up the mass of the doubled sword-scabbard.}\label{fig:1D_line_doubling}
\end{figure}

\begin{proof}
\noindent {\bf Overview of construction.}
The binary expansion of $n \in \mathbb{N}$ is encoded as a horizontal line, denoted $x$, of $\ell = \lfloor \log_2 n \rfloor +1$  adjacent binary nubot monomers (each in state 0 or 1) with neighbouring monomers bound by rigid bonds, placed somewhere on the hexagonal grid.  First, the leftmost of these monomers triggers the growth of a constant sized (length 1) sword and scabbard structure. Then an iterated doubling process begins, that happens exactly $\ell $ times and will result in a sword-and-scabbard of length~$n$ (and height~$5$). At step $i$ of doubling, $1 \leq i \leq \ell$, the leftmost of the input monomers~$x_i$ (from~$x$) is ``read'', and then deleted. If $x_i = 0$ then there will be  a doubling of the length of the sword-and-scabbard, else if $x_i = 1$  there will be  a doubling of the length of the sword-and-scabbard with the addition of one extra monomer. It is straightforward to check that this doubling algorithm finishes with a length $n$ object after $\ell$ rounds.  After the final doubling step, a synchronization occurs, and then $\leq 4n$ of the monomers are deleted (in parallel) in such a way that an $n \times 1$ line remains. All that remains is to show the details of how each doubling step works.

\noindent {\bf Construction details.} Figure~\ref{fig:1D_line_doubling} describes the doubling process in detail: at iteration $i$ of doubling assume that (a) we read an input bit $0$, and that (b) we have a  sword-and-scabbard structure of length $m$ (and height 5). Since the input bit is $0$ we want to double the length to $2m$. As shown in Figure~\ref{fig:1D_line_doubling}(1), we begin with the sword sheathed in the scabbard. We next describe a biased (or ratcheted) random walk process that will ultimately result in the sword being withdrawn all the way to the {\em hook}, giving a structure of length~$2m$.
Via agitation, the sword may be unsheathed by moving out (to the left) of the scabbard, or by the scabbard  moving (to the right) from the sword, although, because of the hook the sword can never be completely withdrawn and hence the two components will never drift apart.\footnote{Besides preserving correctness of the construction, the hook is a safety feature, and hence the sword is merely decorative.}  The withdrawing of the sword is a random walk process with both the sword and scabbard agitating left-right. While this is happening, each monomer---at unit rate, conditioned on that monomer being unsheathed---on the top row of the sword tries to attach a new monomer above. Any such attachment event that succeeds acts as a {\em ratchet} that biases the random walk process in the forward direction. Also, as the sword is unsheathed each unsheathed sword monomer at the bottom of the sword attaches---at unit rate, conditioned on that monomer being unsheathed---a monomer below, and each monomer on the top  (respectively, bottom) horizontal row of the scabbard tries to attach a monomer below (respectively, above) it. These monomers can also serve as ratchets (although in our time analysis below we ignore them which serves only to slow down the analysis). Eventually the sword is completely withdrawn to the hook, and ratcheted at that position, so further agitations do not change the structure.

At this point we are done with the doubling step, and the sword and scabbard reconfigure themselves to prepare for the next doubling (or deletion of monomers if we are done). Figure~\ref{fig:1D_line_doubling}(6)--(9) gives the details. The attachment of new monomers results in 4 new horizontal line segments, each of length $m-1$. Each segment is built in the same way as used for  the synchronization technique shown in Section~\ref{sec:sync}, Figure~\ref{fig:sync}; specifically the bonds are initially formed as rigid, and then transition to flexible in such a way that the line segment (or ``synchronization'' row) is free to agitate relative to its ``backbone'' row only when exactly all $m$ bonds have formed. The line agitates left and right and is then synchronized (or locked into place, see Figure~\ref{fig:sync}) causing all $m$ monomers on the line to change state to ``done''. When the two new line segments that attached to the bottom and top of the sword are both done their rightmost monomers each bind to the scabbard with a rigid bond (as shown in Figure~\ref{fig:1D_line_doubling}(8)) and delete their bonds to the sword (Figure~\ref{fig:1D_line_doubling}(9)) (note that the rightmost of the latter kind of bonds is not deleted until after binding to the scabbard which ensures the entire structure remains connected at all times; also before the leftmost bond on the bottom is deleted a new hook is formed which prevents the new sword leaving the new scabbard prematurely). In a similar process,  the two new line segments that are attached to the scabbard form a new hook, bind themselves to the sword, and then release themselves from the scabbard. We are new ready for the next stage of doubling.

The previous description assumed that the input bit is $0$. If the input bit is instead  $1$ then after doubling both the sword and scabbard are increased in length by 1 monomer (immediately before forming the hook on the new scabbard).

After the final doubling stage then $O(n)$ monomers need to be deleted to leave an $n\times 1$ line of rigidly bound monomers (the goal is to build a line) without having monomers drift away (so as not to violate the space bound).  This is relatively straightforward to achieve: After the final doubling step, a synchronization occurs along the sword, and another along the inside of the scabbard. Then these synchronisation rows signal that the bond structure of all monomers should change to make fully connected $n \times 5$ rectangle, which then changes to become a ``comb'' with a horizontal rigidly connected length $n$ line on top, and 4 ``tooth'' monomers---with no horizontal bonds--- hanging vertically from each top monomer. Using monomer deletion rules, each tooth can then delete itself from bottom to top in 4 steps.  This comb is composed of the topmost row   rigidly binds to the (inside of the scabbard) below and the sword above, and then $\leq 4n $ of the monomers are deleted (in parallel) in such a way that an $n \times 1$ line remains.

\noindent {\bf Expected time analysis.}
Lemma~\ref{lem:1D_line_doubling_time}  states the expected time for a single doubling event: a length~$m$ sword is fully withdrawn to the hook, and locked into place, from a length $m$ scabbard in expected time $O(m^{1/3} \ln m)$.

Between each doubling event there is a reconfiguration of the sword and scabbard. Each reconfiguration invokes a constant number of synchronizations which, via Lemma~\ref{lem:sync_agitate}, take expected time $O(\log m)$ each. Changing of the bond structure also takes place in $O(\log m)$ expected time since each of the four new line segments change their bonds independently, and within a line segment all bond changes (expect for a constant number) occur independently and in parallel.

There are $\ell = \lfloor \log_2 n \rfloor +1$ doubling plus reconfiguration events. By Lemma~\ref{lem:1D_line_doubling_time}, and noting that the length of the sword and scabbard structure during the $k$'th doubling event is $m = \Theta(2^k)$, each doubling event takes time at most $c (2^k)^{1/3} \ln 2^k$ on the $k$'th event for some constant $c$. Then the total expected time is upper bounded by the geometric series
\begin{eqnarray*}
  \sum_{k=0}^{\ell-1} c(2^k)^{1/3} \ln 2^k
  &<&
  c \,\ell \sum_{k=0}^{\ell-1} (2^{1/3})^k
  \\&=& c \,\ell \frac{1 - (2^{1/3})^\ell}{1 - 2^{1/3}}
  \\&=& c \,\ell \cdot O((2^{1/3})^\ell)
  \\&=& O(n^{1/3} \log n).
\end{eqnarray*}
\end{proof}

\subsection{Line length-doubling analysis}
The following lemma is used in the proof of Theorem~\ref{thm:1D_line} and states that, starting from length~$m$, one ``length-doubling'' stage of the line construction completes in expected time $O(m^{1/3} \ln m)$. Intuitively, the proof  shows that the rapid agitation process is a random walk that quickly exposes a large portion of the sword, to which a monomer quickly attaches. This attachment irreversibly  ``ratchets''  the random walk forward, preventing it from walking backwards beyond the attachment position. Eventually the process finishes with the sword completely withdrawn and locked into the withdrawn position. 

\begin{lemma}\label{lem:1D_line_doubling_time}
The expected time for one line-doubling stage (doubling the length) of a length~$m$ sword and scabbard is $O(m^{1/3} \ln m)$.
\end{lemma}

\begin{proof}
  Each stage of the  line construction starts with the sword completely inside the scabbard.
  Any monomer of the sword outside of the scabbard creates a ``blocking monomer'' to its north (top) at constant rate $1$, with a rigid bond, so that no part of the sword to the left of any blocking monomer (in particular, the rightmost blocking monomer) can re-enter the scabbard.
  The stage is completed when the sword is completely out of the scabbard (to the hook) and the rightmost monomer in the sword creates a blocking monomer above it.
  The sword and the scabbard are both undergoing agitation, but for simplicity we may imagine the sword fixed at the origin, and the scabbard agitating relative to the sword at rate $\geq 2m$. %
  We also imagine that the horizontal grid positions on the sword are labeled from left to right by the integers $1,\ldots ,m$ in that order.  In the absence of any blocking monomers, the scabbard undergoes an unbiased random walk on the sequence $1,\ldots ,m$, where the current integer is the position of the left end of the scabbard on the sword.
  A blocking monomer at position $i$ on the sword introduces a reflecting barrier in this walk that from that point on confines the walk to the sequence $i,i+1,\ldots,m$.

  At any time in the process, define the ``ratchet'' to be the rightmost blocking monomer on the sword.
  Let $p \in \{1,\ldots,m\}$ be the position of the ratchet on the sword at some time.
  Let $k = m^{1/3} \ln m$.
  We will show that the expected time for the position of the ratchet to move to the right to relative position at least $p + k$ (i.e., to move right by at least distance~$k$) is $O(\frac{\ln^2 m}{m^{1/3}})$.

  Since this motion of distance $m^{1/3} \ln m$ must happen $m^{2/3} / \ln m$ times for the position of the ratchet to move by $m$ (after which the process is complete), by linearity of expectation, the entire process completes in time
  $\frac{m^{2/3}}{\ln m} \cdot O(\frac{\ln^2 m}{m^{1/3}}) = O(m^{1/3} \ln m)$.

  We first consider the expected time for the $3k$ positions of the sword immediately to the right of the ratchet to become unsheathed.
  We will focus on the three length-$k$ intervals to the right of the ratchet, referring to them as the left, middle, and right intervals.
  In the worst case, they all start out sheathed, and in this case, the expected number of steps for the random walk to move the scabbard right by $3k$ distance is at most $2(3k)^2$.
  Since the agitation rate of the scabbard is $\geq 2m$, this corresponds to expected time at most $\frac{(3k)^2}{m}$.

  We want to upper bound the time at which a new ratchet attaches at least $k$ positions to the right of the current ratchet.
  Any monomer attachment in the middle length-$k$ interval achieves this, so we focus on this event.
  Let $T_a$ be the random variable representing the time for an attachment to occur above the middle interval.
  Our goal is to show
  $E[T_a] = O(\frac{\ln^2 m}{m^{1/3}}).$

  We can bound this expected time by the expected time for a slower ``bounding'' process, in which no attachments are allowed until all three intervals are unsheathed, and in which only attachments on the top of the middle (length $k$) interval are allowed.
  During the time that the entire middle interval is unsheathed, the rate of attachments to the top of this length-$k$ interval is $k$. Hence the expected time for an attachment, conditioned on the entire middle interval being unsheathed, is~$\frac{1}{k}.$
  The middle interval remains completely unsheathed so long as the agitation has not re-sheathed the entirety of the rightmost interval.
  This re-sheathing, conditioned on sword being already unsheathed to exactly position $3k$, takes expected time~$\leq \frac{k^2}{m}$.

  If part of the middle interval becomes re-sheathed, then in our bounding process, (1)~attachments are disallowed until again the entire rightmost interval is unsheathed, and (2)~the entire sword is instantaneously  sheathed (the scabbard is immediately moved as far left as it can go, lining up the ratchet against the scabbard). (Note that the time to unsheath all the intervals from any position is upper bounded by the worst case of unsheathing the entire sword.)

  We can therefore break this bounding process up into epochs, summarized as follows.
  In each epoch, the sword starts with all three intervals completely sheathed.
  The system then undergoes random agitation until the point immediately after unsheathing all three intervals.
  Then attachments in the middle interval are allowed, until agitation re-sheaths the rightmost monomer of the middle interval, at which point all three intervals are immediately re-sheathed, and the next epoch begins.
  The process halts upon the first attachment in the middle interval.
  We write $E_i$ for the event that the attachment happens during the $i$'th epoch.

  We derive upper bounds on $\Exp[T_a | E_i]$.
  Note that for event $E_i$ to occur, we have precisely~$i$ random walks of length~$3k$ (to expose all three intervals $i$ times) and $i-1$ random walks of length~$k$ (to re-sheath the entire right interval $i-1$ times, interleaved between the random walks of length~$3k$), and one attachment event to a length $k$ interval.
  Therefore 
\begin{equation}\label{eq:exptimeattachepoch_i}
  \Exp[T_a | E_i] \leq  i \cdot \frac{(3k)^2}{m} + (i-1) \cdot \frac{k^2}{m} + \frac{1}{k} \leq \frac{1}{k} + \frac{10 i k^2}{m} 
\end{equation}

  Next we lower bound $\Pr[E_1]$, i.e., the probability that an attachment occurs before the middle interval can be re-sheathed.
  Let $c = \ln^2 m$.
  Let $T_k$ be the random variable representing the time for a continuous time unbiased random walk with rate $m$ to move distance $k$.
  By Lemma~\ref{lem-random-walk-tail}, $\Pr[T_k \leq \frac{k^2}{2 m c}] \leq \frac{k^2}{c} e^{-c / 2} + 0.86^{k^2/c}.$

  The time for the middle interval to attach a blocking monomer is an exponential random variable $T'_a$ with rate $k$, given that the middle interval remains unsheathed.
  Hence, $\Pr[T'_a \geq \frac{k^2}{2 c m}] = e^{-k \frac{k^2}{2 c m}} = e^{-\frac{k^3}{2 c m}}$.

  If $T'_a < \frac{k^2}{2 c m}$ and $T_k > \frac{k^2}{2 c m}$, then $E_1$ occurs, so by the union bound and the above two probability bounds, $\Pr[\neg E_1] \leq e^{-\frac{k^3}{2 c m}} + \frac{k^2}{c} e^{-c / 2} + 0.86^{k^2/c}$.
  Recall that $c=\ln^2 m$ and $k=m^{1/3} \ln m$, which implies that
  \begin{eqnarray*}
  \Pr[\neg E_1]
  &\leq&
  e^{-\frac{(m^{1/3} \ln m)^3}{2 m \ln^2 m}} + \frac{(m^{1/3} \ln m)^2}{\ln^2 m} e^{- \frac{1}{2} \ln^2 m} + 0.86^{(m^{1/3} \ln m)^2 / \ln^2 m}
  \\&=&
  e^{-\frac{1}{2} \ln m} + m^{2/3} e^{- \frac{1}{2} \ln^2 m} + 0.86^{m^{2/3}}
  \\&=&
  e^{-\frac{1}{2} \ln m} + m^{2/3} (e^{\ln m})^{- \frac{1}{2} \ln m} + 0.86^{m^{2/3}}
  \\&=&
  e^{-\frac{1}{2} \ln m} + m^{2/3 - \frac{1}{2} \ln m} + 0.86^{m^{2/3}}
  \end{eqnarray*}
  Since all three terms go to 0 as $m \to\infty$, for sufficiently large $m$ the right side is at most $\frac{1}{4}$.

We use similar reasoning to derive bounds on $\Pr[E_i]$ for $i > 1$:  for $E_i$ to occur, we must have either $T'_a \geq \frac{k^2}{c m}$ or $T_k \leq \frac{k^2}{c}$ occur $i-1$ times in a row independently (i.e., we must have the events ($\neg E_1$) AND ($\neg E_2$ conditioned on $\neg E_1$) AND $\ldots$ AND ($\neg E_{i-1}$ conditioned on $\neg E_{i-2}$ AND  $\neg E_{i-3}\ldots$) in order for the first $i-1$ re-sheathings to occur before an attachment can occur in the middle interval).
  Since these are independent
\begin{equation}\label{eq:probattachepoch_i}
 \Pr[E_i] \leq \frac{1}{4^{i-1}} = \frac{1}{2^i}
\end{equation}

We combine the bound on $\Exp[T_a | E_i]$ (Equation~\eqref{eq:exptimeattachepoch_i}) with the bound on $\Pr[E_i]$ (Equation~\eqref{eq:probattachepoch_i}), to get
  \begin{eqnarray*}
    \Exp[T_a]
    &=&
    \sum_{i=1}^\infty \Exp[T_a | E_i] \cdot \Pr[E_i]
    \\&\leq&
    \sum_{i=1}^\infty \left( \frac{1}{k} + \frac{10 i k^2}{m} \right) \cdot \Pr[E_i]
    \\&=&
    \frac{1}{k} + \frac{10 k^2}{m} \sum_{i=1}^\infty i \cdot \Pr[E_i]
    \\&\leq&
    \frac{1}{k} + \frac{10 k^2}{m} \sum_{i=1}^\infty i \cdot \frac{1}{2^i}
    \\&=&
    \frac{1}{k} + \frac{20 k^2}{m} \ \ \ \ \text{since the sum converges to 2}
    \\&=&
    \frac{1}{m^{1/3} \ln m} + \frac{20 (m^{1/3} \ln m)^2}{m}
    \\&=&
    \frac{1}{m^{1/3} \ln m} + 20 \frac{\ln^2 m}{m^{1/3}}
    \\&<&
    21 \frac{\ln^2 m}{m^{1/3}}.
  \end{eqnarray*}

  This is a bound on the expected time for the ratchet to move right by distance $k = m^{1/3} \ln m$.
  Since this must occur $\frac{m}{m^{1/3} \ln m} = \frac{m^{2/3}}{\ln m}$ times for the ratchet to move the complete distance $m$, by linearity of expectation the expected time to move distance $m$ is at most $21 \frac{\ln^2 m}{m^{1/3}}  \frac{m^{2/3}}{\ln m} = 21 m^{1/3} \ln m$.
\end{proof}

The following  technical lemma bounds the probability that a continuous-time random walk takes much longer than its expected time to reach a certain distance from the starting point.
\begin{lemma}\label{lem-random-walk-tail}
  Let $k\in\Z^+$, let $c \in \R^+$, let $m > 0$ and let $T_k$ be the random variable describing the amount of time taken by a continuous-time random walk on $\Z$ with rate $m$ to reach the value~$k$ for the first time, starting from 0.
  Then $\Pr[T_k \leq \frac{k^2}{2 m c}] \leq \frac{k^2}{c} e^{-c / 2} + 0.86^{k^2/c}.$
\end{lemma}
\begin{proof}
  For all $i \in \Z^+$, define $B_i$ to be a random variable  with $\Pr[B_i = +1] = \Pr[B_i = -1] = \frac{1}{2}$.
  For all $j \in \Z^+$, define $S_j = \sum_{i=1}^j B_i$.
  Let $t_k$ be the number of steps needed for the random walk to reach $k$.
  For any $\tau \in \Z^+$, the condition that $t_k \leq \tau$ is equivalent to the condition that $(\exists j \in \range{\tau})\ S_j \geq k$.

  Then by the union bound, $\Pr[t_k \leq \frac{k^2}{c}] \leq \sum_{j=1}^{k^2 / c} \Pr[S_j \geq k]$.
  We use the following form of the Chernoff bound: If $B_i$ is a random variable with $\Pr[B_i = +1] = \Pr[B_i = -1] = \frac{1}{2}$, then for any $j,t$, $\Pr[\sum_{i=1}^j B_i \geq t \cdot \sqrt{j}] \leq e^{-t^2 / 2}$.
  By this bound, $\Pr[S_j \geq k] = \Pr[S_j \geq \frac{k}{\sqrt{j}} \sqrt{j}] \leq e^{-k^2/(2 j)}$.
  Therefore
  \begin{equation}\label{ineq-1}
    \Pr\left[t_k \leq \frac{k^2}{c}\right]
    \leq \sum_{j=1}^{k^2 / c} e^{-k^2/(2 j)}
    < \sum_{j=1}^{k^2 / c} e^{-k^2/(2 k^2 / c)}
    = \frac{k^2}{c} e^{-c / 2}.
  \end{equation}

  We now bound $\Pr[T_k \leq \frac{k^2}{2 m c} | t_k \geq \frac{k^2}{c}]$.
  In the worst case, $t_k = \frac{k^2}{c}$, so we make this assumption.
  Under this assumption, the event $T_k \leq \frac{k^2}{2 m c}$ is equivalent to the event that a sum of $\frac{k^2}{c}$ exponential random variables, each with rate $m$, takes value at most $\frac{k^2}{2mc}$.

  Recall that the moment-generating function of an exponential random variable $X_i$ with rate $m$ is $M_{X_i}(\theta) = \Exp[e^{\theta X_i}] = \frac{m}{m-\theta}$, defined whenever $|\theta| < m$.
  Then if $X = \sum_{i=1}^{k^2/c} X_i$ is a sum of independent exponential random variables, each with rate $m$, we have
  $$
    \Exp[e^{\theta X}]
    = \Exp[e^{\theta \sum_{i=1}^{k^2/c} X_i}]
    = \Exp\left[\prod_{i=1}^{k^2/c} e^{\theta X_i} \right]
    = \prod_{i=1}^{k^2/c} \Exp[e^{\theta X_i}]
    = \left( \frac{m}{m-\theta} \right)^{k^2/c}.
  $$
  Therefore, if $\theta < 0$, by Markov's inequality,
  $$
    \Pr\left[X \leq \frac{k^2}{2mc}\right]
    = \Pr[e^{\theta X} \geq e^{\theta k^2 / (2mc)}]
    \leq \frac{\Exp[e^{\theta X}]}{e^{\theta k^2 / (2mc)}}
    = \frac{\left( \frac{m}{m-\theta} \right)^{k^2/c}}{e^{\theta k^2 / (2mc)}}.
  $$
  Letting $\theta = -m/2$, the above expression becomes
  \begin{equation}\label{ineq-2}
    \frac{\left( \frac{2}{3} \right)^{k^2/c}}{e^{- k^2 /(4c)}}
    = \left( \frac{2 e^{1/4}}{3} \right)^{k^2/c}
    < 0.86^{k^2/c}.
  \end{equation}

  Combining the bounds~\eqref{ineq-1} and~\eqref{ineq-2} with the union bound, we have
  $\Pr[T_k \leq \frac{k^2}{2 m c}] \leq \frac{k^2}{c} e^{-c / 2} + 0.86^{k^2/c}.$
\end{proof}

%% file: sections/appendix-movement.tex
\section{Movement rule definition}\label{sec:movement_definition}

The nubot movement rule is not used in the constructions in this paper, but we include its definition here since it is part of the full nubot model~\cite{nubots,nubotsDNA19}.

From Section~\ref{sec:defs}, a rule is of the form $(s1, s2, b, \overrightarrow{u}) \rightarrow (s1', s2', b', \overrightarrow{u}')$. For a \emph{movement} rule, $\overrightarrow{u} \neq \overrightarrow{u}'$. Also, it must be the case that $d(\overrightarrow{u},\overrightarrow{u}') = 1$, where $d(u,v)$ is Manhattan distance on the triangular grid, and  $s1,s2,s1',s2' \in S \setminus \{ \mathsf{empty} \}$. If we fix $\overrightarrow{u} \in \mathcal{D}$, then there are two $\overrightarrow{u}' \in \mathcal{D}$ that satisfy $d(\overrightarrow{u},\overrightarrow{u}') = 1$.  A movement rule is applicable if it can be applied both (i) locally and (ii) globally, as follows. 

(i) Locally, the pair of monomers should be in state $s1, s2$, share bond $b$ and have orientation~$\overrightarrow{u}$ of $s2$ relative to $s1$. Then,  one of the two monomers is  chosen nondeterministically to be the {\em base} (that remains stationary), the other is the {\em arm} (that moves).  If the~$s2$ monomer, denoted $X$, is chosen as the arm then $X$ moves from its current position~$p(X)$ to a new position $p(X) - \overrightarrow{u} + \overrightarrow{u}'$. After this movement (and potential state change), 
$\overrightarrow{u}'$ is the relative position of the~$s2'$ monomer to the~$s1'$ monomer, as illustrated in Figure~\ref{fig:model}(b). Analogously, if the $s1$ monomer, $Y$, is chosen as the arm then~$Y$ moves from  $p(Y)$ to  $p(Y) + \overrightarrow{u} - \overrightarrow{u}'$. Again, $\overrightarrow{u}'$ is the relative position of the~$s2'$ monomer to the~$s1'$ monomer.  
Bonds and states may change during the movement. 

(ii) Globally, the movement rule may push, or pull other monomers, or if it can do neither then it is not applicable. This is formalized as follows, see~\cite{nubots} or~\cite{nubotsDNA19} for examples.
Using the definition of agitation set, Definition~\ref{def:agitationset}, we define the {\em movable set} $\mathcal{M}(\config, A, B, \vec{v})$ for a pair of monomers $A, B$, unit vector $\vec{v}$ and configuration~$C$.

\begin{definition}[Movable set]\label{def:move}
Let $\config$ be a configuration containing adjacent monomers $A,B$,  let $\vec{v} \in \mathcal{D}$ be a unit vector, and let~$\config'$ be the same configuration as~$\config$ except that~$\config'$ omits  any bond between $A$ and $B$. The \emph{movable set} $\mathcal{M}(\config, A,B,\vec{v})$ is defined to be the agitation set $\mathcal{A}(\config', A, \vec{v})$ if $B \not\in \mathcal{A}(\config', A, \vec{v})$, and  the empty set otherwise. 
\end{definition}

If $\mathcal{M}(\config, A, B, \overrightarrow{v}) \neq \{ \}$, then the movement where $A$ is the arm (which should be translated  by~$\overrightarrow{v}$) and $B$ is the base (which should not be translated) is applied as follows: (1) the movable set $\mathcal{M}(\config, A, B, \overrightarrow{v})$ moves unit distance along $\overrightarrow{v}$;  (2) the states of, and the bond between, $A$ and $B$ are updated according to the rule;  (3) the states of all the  monomers besides $A$ and $B$ remain unchanged and pairwise bonds remain intact (although monomer positions and flexible/null bond orientations may change). 
If $\mathcal{M}(\config, A, B, \overrightarrow{v}) = \{ \}$, the movement rule is inapplicable  (the rule is ``blocked'' and thus $A$ is prevented from translating).